\documentclass[11pt]{article}

\usepackage[english]{babel}
\usepackage[utf8x]{inputenc}
\usepackage[margin=1in]{geometry}
\usepackage{amsmath}
\usepackage{amsfonts}
\usepackage{soul}
\usepackage{amssymb}
\usepackage{amsthm}
\usepackage{rotating}
\usepackage{amsbsy}
\usepackage{graphicx,  ucs}
\usepackage{tikz}
\usepackage{listings}
\usepackage{enumitem}
\usepackage{hyperref}
\usepackage{multirow}
\usepackage{array}
\usepackage[compact]{titlesec}
\usepackage{cmap}
\usepackage[T1]{fontenc}
\usepackage{bm}
\usepackage{comment}

\def\compactify{\itemsep=0pt \topsep=0pt \partopsep=0pt \parsep=0pt}
\let\latexusecounter=\usecounter
\newenvironment{CompactEnumerate}
  {\def\usecounter{\compactify\latexusecounter}
   \begin{enumerate}}
  {\end{enumerate}\let\usecounter=\latexusecounter}
\newenvironment{Itemize}
{\begin{itemize}
\setlength{\itemsep}{0pt}
\setlength{\topsep}{0pt}
\setlength{\partopsep}{0 in}
\setlength{\parskip}{0 pt}}
{\end{itemize}}

\newcommand{\Domain}{\mathcal{D}}
\newcommand{\Support}{\mathcal{X}}

\newcommand{\Workers}{\mathcal{N}}
\newcommand{\Truth}{\mathcal{T}}
\newcommand{\xw}{\vec{x}_{\Workers}}
\newcommand{\xt}{\vec{x}_{\Truth}}

\newcommand{\Exp}{\mathbb{E}}

\newcommand{\com}[1]{}

\newtheorem{theorem}{Theorem}

\newtheorem{definition}{Definition}
\newtheorem{corollary}{Corollary}
\newtheorem{lemma}{Lemma}

\def\poly{\mathrm{poly}}
\def\eps{\varepsilon}

\def\Prob{\mathbb{P}}
\def\Exp{\mathbb{E}}

\def\Cond{\textsc{Cond}}

\def\NP{\textsc{NP}}

\newcommand{\reals}{\mathbb{R}}
\newcommand{\nats}{\mathbb{N}}

\def\abs#1{\left|#1\right|}

\renewcommand{\vec}{\mathbf}

\definecolor{vergreen}{RGB}{0,85,2}
\definecolor{myvergreen}{RGB}{0,140,3}
\definecolor{provorange}{RGB}{85,34,0}
\definecolor{inputblue}{RGB}{5,13,111}
\definecolor{noapred}{RGB}{116,3,3}
\definecolor{classesblue}{RGB}{9,49,146}

\definecolor{secinhead}{RGB}{249,196,95}

\definecolor{lgray}{gray}{0.8}

\begin{document}
\title{Certified Computation from Unreliable Datasets}

\author{Themis Gouleakis \thanks{tgoule@mit.edu}\\
MIT
\and Christos Tzamos \thanks{tzamos@mit.edu}\\
Microsoft Research
\and Manolis Zampetakis \thanks{mzampet@mit.edu}\\
  MIT
}
\date{}
\maketitle

\begin{abstract}
  A wide range of learning tasks require human input in labeling massive data.
  The collected data though are usually low quality and contain inaccuracies and
  errors. As a result, modern science and business face the problem of learning
  from unreliable data sets.

  In this work, we provide a generic approach that is based on
  \textit{verification} of only few records of the data set to guarantee high
  quality learning outcomes for various optimization objectives. Our method,
  identifies small sets of critical records and verifies their validity. We show
  that many problems only need $\poly(1/\eps)$ verifications, to ensure that the
  output of the computation is at most a factor of $(1 \pm \eps)$ away from the
  truth. For any given instance, we provide an \textit{instance optimal}
  solution that verifies the minimum possible number of records to approximately
  certify correctness. Then using this instance optimal formulation of the
  problem we prove our main result: "every function that satisfies some
  Lipschitz continuity condition can be certified with a small number of
  verifications". We show that the required Lipschitz continuity condition is
  satisfied even by some $\NP$-complete problems, which illustrates the
  generality and importance of this theorem.

  In case this certification step fails, an invalid record will be identified.
  Removing these records and repeating until success, guarantees that the result
  will be accurate and will depend only on the verified records. Surprisingly,
  as we show, for several computation tasks more efficient methods are possible.
  These methods always guarantee that the produced result is not affected by the
  invalid records, since any invalid record that affects the output will be
  detected and verified.
\end{abstract}

  \section{Introduction}
\label{s:intro}

  Modern science and business involve using large amounts of data to perform various computational or learning tasks. The data required by a
particular research group or enterprise usually contain \textit{errors and inaccuracies} because of the following reasons:
\begin{CompactEnumerate}
  \item[-] the validity of the data changes dynamically. For example data involving home locations of customers or employees are
           not constant over time. Hence a set of data collected in a particular time frame probably is not going to remain valid for the
           future,
  \item[-] the data might be provided by other entities or collected online from a source that has no certification for their validity. For
           example data collected from crowdsourcing environments.
\end{CompactEnumerate}

\noindent We call a set of data with the property that only a subset of them is valid an \textit{unreliable data set}. The goal of this paper
is to develop theoretically established methods that lead to \textit{certified computation} over such data sets. Towards this goal we
assume that we have the  ability to \textit{verify} the validity of a record in our data set. Usually this verification process is costly and hence
it doesn't make sense to verify all the records in our data set every time we want to compute a function on them. On the other hand if the
majority or the most important part of the data are invalid then trying to find a valid subset could lead to essentially verifying the entire data set. In our work we introduce the concept of \textit{learning with certification}, in which we can distinguish between the following scenarios
\begin{CompactEnumerate}
  \item the value of the function that we computed on the unreliable data set is close to the value of the function computed on the valid
           subset of our data set,
  \item there exists at least one invalid record that could dramatically charge the value of the function that we want to compute.
\end{CompactEnumerate}
by verifying only a small number of records.
\smallskip

\noindent \textbf{Computations in Crowdsourcing.} Crowdsourcing \cite{DoanRH2011} is a popular instantiation of an unreliable data set,
where records are provided by a very large number of workers. These workers may need to put significant effort to extract high quality data
and without the right incentives they might choose not to do so giving, as a result, very noisy and unreliable reports. Experimental evidence
\cite{KazaiKKM2011, VuurensVE2011, WaisLCFGLMS2010} suggests that There are a large number of examples where crowdsourcing fails in practice
because of the unreliability of the data that it produces. An anecdotal failure of crowdsourcing is the example of Walmart's mechanism that
made the famous rapper Pitbull travel to the remote island of Kodiak, Alaska, see e.g., \cite{Pitbull12}. In 2012, Walmart asked their
customers to vote, through Facebook, their favorite local store. The store with the most votes would host a promotional performance by Pitbull.
Probably as a mean joke, a handful of people organized an \#ExilePitbull campaign, inviting Facebook users to vote for the most remote Walmart
store, at Kodiak. The campaign went viral and Pitbull performed at Kodiak, in July 2012. While the objective of Walmart was to learn the
location that would maximize attendance to the concert, the resulting outcome was terribly off because the incentives of the workers were
misaligned.

  Our work is motivated by these observations and aims through the use of verification to provide a generic approach that
guarantees high quality learning outcomes. Verification can be implemented either directly, in tasks such as peer grading,
by having an expert regrade the assignment, or indirectly, e.g. in the Walmart example by verifying the locations of the
voters. The main challenge in our framework it to minimize such verifications since they can be very costly.

\subsection{Our Model and Results}

  A data set is a set of records $\Workers$. The set $\Workers$ may contain, apart from the valid subset of records $\Truth$, a set of records
$\Workers \setminus \Truth$ that are invalid due to reasons that we described earlier. But how much does the presence of these invalid records
affect the output of the computation? The answer to this question depends on the number but also on the \textit{importance} of the invalid
records, where the importance depends on the specific computation task that we want to run.

  In order to assess whether the computed output is accurate, we can verify the validity of some of the records. Our goal is to verify as few
reports as possible and eventually be confident that the output of the computation is accurate. At this point we have to define a measure to
evaluate the accuracy of an output. Ideally, an accurate output is the output that we would get if all the records were valid. Such a benchmark,
however, is impossible to achieve as the correct values of the invalid records are unobservable. We instead focus on a simpler benchmark. We want
to decide whether given an unreliable data set the output of the computation based on $\Workers$ is close to the output of the computation based
only on $\Truth$. That is, if we could see which records are invalid and perform the computation task after discarding them, would the output
of the computation be close to the current value?

\paragraph{Certification Schemes} A positive answer to the question above is called \textit{certification} of the computation task based on the
unreliable data set $\Workers$. A negative answer is a witness that at least one record in $\Workers$ is invalid. Our first goal of this paper is to provide certification schemes for general computation
tasks that verify only a small number of records and can distinguish between these two cases.

\begin{figure}[!h] \label{fig:immunity}
  \centering
  {\includegraphics[clip, trim=2cm 8cm 2cm 1cm,width=0.45\textwidth]{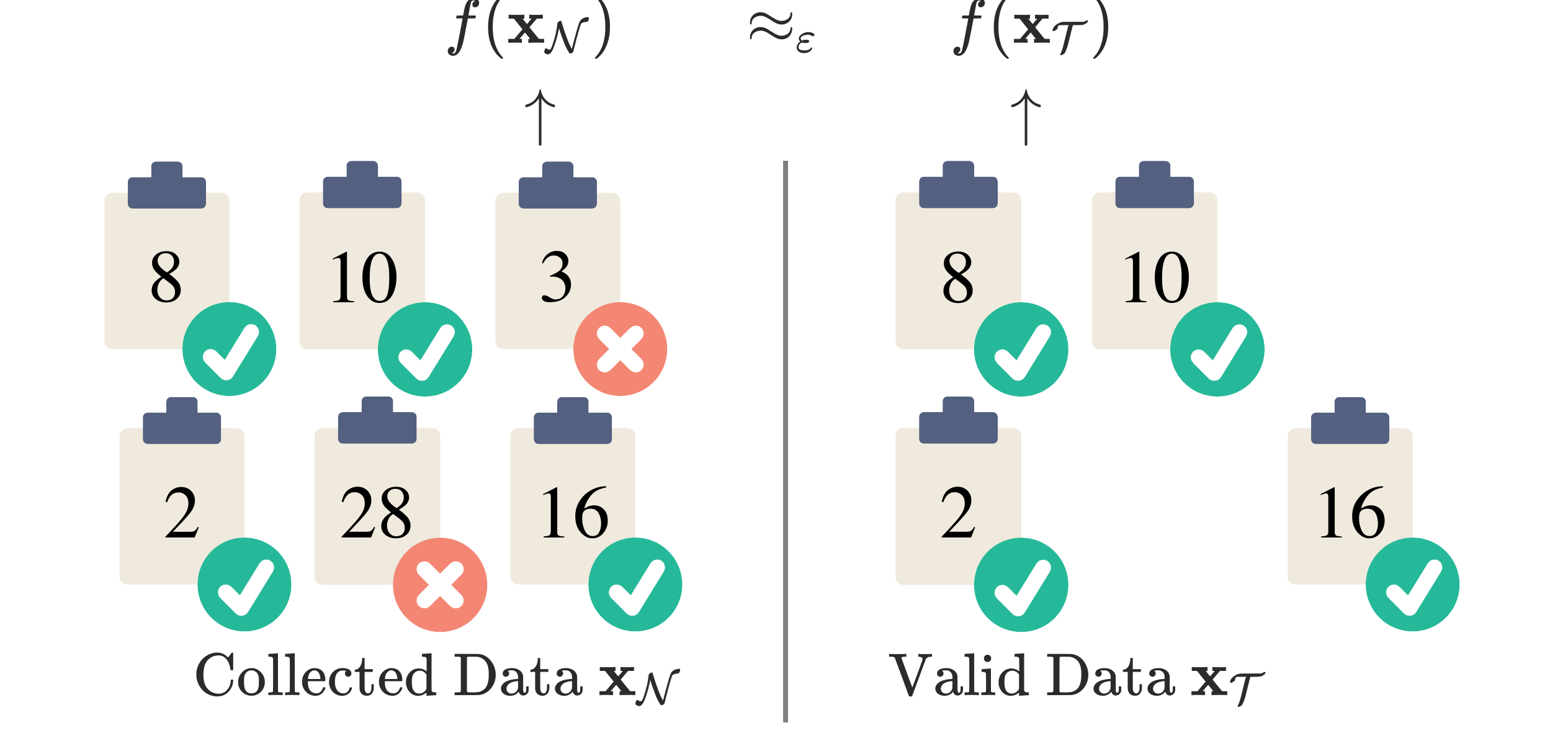}}
  \caption{The property that a certification scheme certifies.}
\end{figure}

As a toy example of these models let us consider the simple function $f(\xw) = \max_{i \in \Workers} x_i$, where we assume that each record is a
real number $x_i \in \reals$. For the certification task we want to check whether $f(\xw) = f(\xt)$ or not. This can be easily done by checking
whether the record $i^* = \arg \max_{i \in \Workers} x_i$ is valid or not.

Not all functions though have such efficient deterministic and exact correction schemes. For several functions, we can obtain randomized
certification schemes that succeed with high probability and certify that the output is close up to a multiplicative factor. Moreover, for some
functions it might not even be possible to efficiently certify them without certifying almost everything. One extreme such example is a threshold
function that is 1 if all records are valid and 0 otherwise, $\mathbb{I}_{\Workers=\Truth}$ where we cannot obtain any meaningful approximation
without verifying $\Omega(n)$ records of $\Workers$.

In Section~\ref{sec:certification}, we provide efficient certification schemes for many different functions. Our results are the following:
\begin{Itemize}
  \item[-] {\bf sum function.} We start by presenting a randomized scheme for certifying the sum of all records that uses only
           $O(\frac 1 \varepsilon)$ verifications to certify correctness up to a multiplicative factor $1 \pm \varepsilon$. This is a very useful
           primitive that can be used in several different tasks: For example for computing the average, we can compute and certify the total sum
           of records and divide by the total number of records which we can also certify as another summation task. Another example is the
           \emph{max-of-sums} function, where as in the Walmart example we presented earlier, agents vote on different categories and the goal is
           to compute the total category that has the maximum number (or sum) of valid records. This can be easily certified by computing the max
           of all sums of records and certifying that this sum is approximately correct.
  \item[-] {\bf functions given by linear programs.} We then study the set of functions  expressible as LPs where the input data
           corresponds to either variables or constraints of the LP. We show that for functions expressible as packing or covering LPs, only
           $O(\frac 1 \varepsilon)$ verifications to certify correctness up to a multiplicative factor $1 \pm \varepsilon$ while for more general
           LPs we provide a deterministic scheme that depends on the dimensionality (number of variables or constraints).
  \item[-] {\bf instance-optimal schemes.} To study more general functions, we devise a linear program that characterizes (up to a constant factor)
           for any given instance the minimum number of verifications needed for approximate certification. We show that even though optimal
           certification schemes may be arbitrarily complex, there are simple schemes that verify records independently that are almost-optimal.
  \item[-] {\bf Main theorem for certification.} Our main most general result is that we can provide explicit solutions to the instance optimal
           linear program for a large class of different objectives that satisfy a $w$-Lipschitz property. We illustrate the flexibility of this
           constraint by showing that even very complex functions that correspond to NP-hard problems satisfy the $w$-Lipschitz property.
           Specifically, using our general theorem, we prove this for the TSP problem and the Steiner tree problems where we show that the
           certification complexity is only $O(\frac 1 \varepsilon)$. These capture settings where agents report their locations in a metric space
           and the goal is to design an optimal tour that visits all of them (TSP) or connecting them in a network by minimizing total cost
           (Steiner tree).
\end{Itemize}

\paragraph{Correction Schemes} Although very useful, the certification process fails when at least one invalid record is found. Naturally the next
question to ask is how we can proceed in order to actually compute the value of the function that we are interested in by throwing away the invalid
records. In a worst case example where all records are invalid, we would need to verify all the agents to complete the correction task. To get a more
meaningful and realistic measure of the verification complexity of a correction task, we carefully define it in terms of a budget $B$. For verification
complexity $V$, we assume that the designer has an initial budget for verifications $B = V$ which decreases as he performs verifications but might
increase every time he finds an invalid record.
The rationale behind the increase is that  verifications of incorrect records lead to their removal from the data set, which makes it more accurate. Therefore, in this model we measure the number of verifications needed to correct one incorrect record.
We distinguish correction schemes
into two models depending on the budget increase.

In the \textit{weak correction} model, the budget increases by $V$ every time an invalid record is found. This means that finding an invalid record allows
us to restart the process from the beginning.

In the \textit{strong correction} model, the budget does not increase but does not decrease either. This means that verification of invalid records is
costless.

Of course, as we said, in the worst case a correction scheme has to verify all the data in the data set which is not realistic. However,  during the correction procedure for specific tasks it would be reasonable to have an upper bound on the number of invalid data that we are willing to verify before dismissing the entire data set for being too corrupted.
In particular, if the correction scheme
succeeds within the verification budgets we have an accurate output to our computation task. Otherwise, we can conclude that our data set is too corrupted and hence we need to collect the data from the beginning.

  Notice that in the example of the $\max$ function, if the certification fails then we can continue by checking the second largest record and so on
until we find a valid record which will give us the value $f(\xt)$ precisely. 
However, strong correction schemes are much harder to obtain than weak ones in general.

  If our computation or learning task has a deterministic certification scheme, e.g. the $\max$ function, it is easy to obtain weak correction
schemes by repeating the certification scheme until success. For randomized schemes though, one needs to be more careful as it is possible that errors
can accumulate. This is easy to fix by requiring that the certification scheme fails with probability at most $1/n$. However, this increases the total
weak-verification complexity by a logarithmic factor.

  In Section~\ref{sec:wcorr}, we prove our main result for weak correction schemes which implies that such an increase is not necessary and it is
possible to obtain weak correction schemes with the same complexity as the underlying certification scheme (up to constant factors). To do this, we run
the certification scheme many times and do not stop the first time it succeeds but continue until the total number of successes is more than the number
of failures. A random walk argument guarantees that this produces the correct answer with constant probability. If the objective function is not
monotone, additional care is needed to get the same guarantee.

  While weak correction schemes with good verification complexity exist for all tasks that we can efficiently certify, strong correction schemes are
more rare. In Section~\ref{sec:scorr}, we show that it is possible to obtain strong correction schemes for the sum function using only
$O(\frac 1 {\varepsilon^2})$ verifications of valid records. Since that many verifications are necessary to get a $1 \pm \varepsilon$ multiplicative
approximation for the sum, this implies a gap between the weak and strong correction models. The gap between them can be arbitrarily large though. As an example, the max-of-sums function we discussed earlier has certification and weak-correction complexity $O(\frac 1 {\varepsilon})$, although it is impossible to get a constant factor approximation in the strong correction model without verifying $\Omega(n)$ valid records.

  Despite the impossibility of obtaining strong correction schemes even for simple functions such as the max-of-sums, we can show that efficient
certification schemes exist for quite general optimization objectives. We prove (Theorem~\ref{thm:sCorrection}) a very interesting and tight connection
of strong correction schemes with \textit{sublinear algorithms that use conditional sampling}~\cite{GouleakisTZ2017}. We can exploit this connection to
directly obtain efficient strong correction schemes. This gives efficient schemes for general optimization tasks such as clustering, minimum spanning tree,
TSP and Steiner tree that capture settings where agent reports lie on some metric space.

\subsection{Related Work}
Our certification task resembles the task of property testing, as formalized in \cite{GGR96},
where one has to decide whether the data has a particular property versus being $\varepsilon$-far from it in some distance metric. In our case, the property we want to test is whether the evaluation of the function on all the collected data is equal to its evaluation on the subset of the valid data only.

  Our correction task is related to a large body of work in statistics on how to
deal with noisy or incomplete datasets. Several methods have been proposed for dealing with missing data. The popular method of imputation
\cite{Rubin:87,book:imputation,book:incomplete:data} corrects the dataset by filling in the missing values using the
maximum likelihood estimates.

In addition, the field of robust statistics \cite{hampel80, huber11} deals with the problem of designing estimators when the dataset contains random or adversarially corrupted datapoints. Several efficient algorithmic results have recently appeared in the context of robust parameter estimation and distribution learning \cite{DKK+16,DKK+17,DKK+18,DKS17,LRV16}. The goal of these works is to learn the parameters of a multidimensional distribution, that belongs to a known parametric family, while a constant fraction of the samples have been adversarially corrupted.
\cite{CharikarSG17} deal with parameter estimation in cases where more than half of the dataset is corrupted and identification is impossible by providing a list of candidate estimates. They show that the correct estimate can be chosen as long as a small ``verified'' set of data is provided. In contrast to \cite{CharikarSG17}, the verification oracle in our model allows us to verify any subset of datapoints but verification is costly.   
\cite{SVC16} consider similar verification access to the dataset in crowdsourced peer grading settings.

Selective verification of datapoints has also been explored in the context of mechanism design. \cite{FotakisTZ16} study mechanisms with verification and achieve truthfulness by solving a task similar to certification in social choice problems.

Finally, another related branch of literature considers the task of correcting datasets through local queries (\cite{JhaR11,BlumLR90,BhattacharyyaGJJRW12,SacksS10,AilonCCSL08,CanonneGR16}). For example, using local queries, \cite{AilonCCSL08} correct datasets to ensure monotonicity and other structural properties. \cite{CanonneGR16} solve similar local correction tasks for noisy probability distributions.

  \section{Model and Preliminaries}
\label{sec:prelim}

\paragraph{Notation} For $m \in \nats$ we denote the set $\{1,\cdots,m\}$ by $[m]$. 
Let $\Workers = [n]$ be the set of all records of the data set and $\Truth\subseteq \Workers$ be the subset of records
that are valid. The set $\Truth$ is unknown to the algorithm. Suppose we are given an input $\xw = (x_1, x_2 , \cdots , x_n)$ 
of length $n$, where every $x_i$ belongs to some set $\Domain$. Let $\vec{x}_{\cal T} = (x_j)_{j\in {\cal T}}$ be a vector 
consisting only of the coordinates of $\vec{x}$ that are in $\Truth$. Our general goal is to approximate the value of a 
symmetric function $f : \Domain^* \rightarrow \mathbb{R}_+$ on input $\vec x_\Truth \in \Domain^*$. Finally, Every input $x_j$
with $j \in \Truth$
is called valid and the rest $\Workers \setminus \Truth$ are called invalid. We consider two different tasks; \emph{certification}
and \emph{correction}.

\smallskip
\noindent In the {\bf certification task}, we count the total number of verifications of records needed to test between
the following two hypotheses:

\begin{enumerate}[label=(H\arabic*)]
    \item \[ f(\xw) \in \left[1 - \eps, \frac{1}{1 - \eps} \right] \cdot f(\xt) \] \label{eq:verH0}
    \item there exists a record $i$ such that $i \notin {\cal T}$. \label{eq:verH1}
\end{enumerate}

\noindent We allow a small probability $\delta$ that the algorithm fails to find a witness, i.e.
\[ \Prob\left( f(\xw) \not\in \left[1-\eps, \frac 1 {1-\eps} \right] \cdot f(\xt) \wedge \text{ no invalid record found} \right) \le \delta \]

\smallskip
\noindent In the {\bf correction task},  
the goal is to always compute an approximation to the correct answer, even when the certification task fails.  
We consider two models for correction the {\bf weak correction} and the {\bf strong correction}.

\smallskip
\noindent In the {\bf weak correction} model after catching an invalid record we are allowed to restart the task and
therefore we do not count the number of verifications that we already used before catching the invalid record. So if we have 
the guarantee that a weak correction scheme uses $v(n , \eps)$ verifications and during the execution of the scheme we find
$k$ invalid records, then the total number of verifications used is at most $(k + 1) \cdot v(n, \eps)$.

\smallskip
\noindent In the {\bf strong correction} model instead of restarting every time we find an invalid record, we just ignore 
the data of this record and we also don't count them in the number of verifications. So if we have the
guarantee that a strong correction scheme uses $v(n , \eps)$ verifications and during the execution of the scheme we find $k$
invalid records, then the total number of samples used is at most $k + v(n, \eps)$.
 
  \section{Certification Schemes for Linear Programs} \label{sec:certification}

  In this section, we present examples of certification schemes for frequently arising problems such as computing the sum 
of values or functions that can be expressed as a linear programs. In the next section we will see the more general statement
about certification schemes for functions that satisfy a general Lipschitz continuity condition.

\subsection{Computing the Sum of Records} \label{ssec:certificationSum}

  One of the most basic certification tasks is computing the sum of the values of the records. For this task, we are given $n$ 
positive real numbers $x_1, x_2, \dots, x_n$ each one comming from a record in our data set. Our goal is to certify whether the
sum of all the records is closed to the sum of the subset of records that are valid, i.e. belong to $\Truth$. 

  More formally, we want to check with probability of failure at most $\delta > 0$ whether 
$\sum_{i \in \Workers} x_i \in \left[1 - \eps, \frac{1}{1 - \eps} \right] \cdot \sum_{i \in \Truth} x_i$ or there is at least 
one record $i$ such that $i \notin {\cal T}$. We show that there exists an efficient certification scheme for this task:

\begin{figure}[!h] \label{fig:certification sum}
  \centering
  {\includegraphics[clip, trim=1cm 17cm 2cm 0.2cm,width=0.65\textwidth]{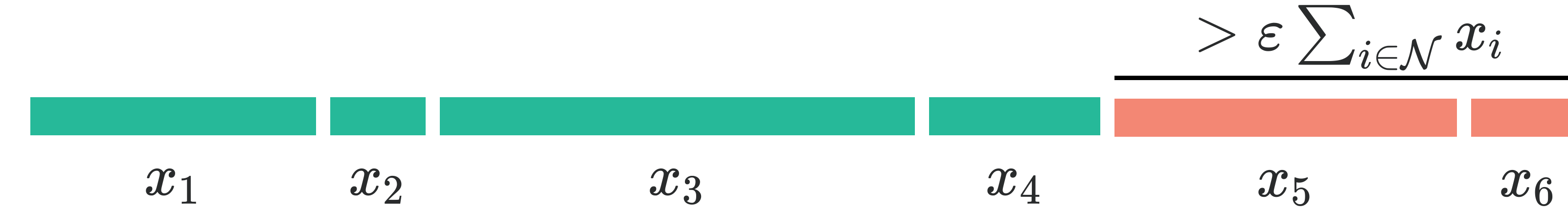}}
  \caption{If the invalid records make up more than $\eps$ fraction of the total sum, there is at least $\eps$ probability that an invalid record is found with a single verification.}
\end{figure}

\begin{lemma}\label{lem:sumV}
Let $x_1, x_2, \dots, x_n \ge 0$ be the values of the records in $\Workers$ and $f(\xw) = \sum_{i \in \Workers} x_i$. Consider 
the probability distribution $p_i = \frac {x_i} {\sum_j x_j}$ which assigns to each record a probability proportional to its 
value $x_i$. Verifying $k=\Theta(\frac{1}{\eps}\log (1/\delta))$ records sampled independently from $p$, guarantees that the 
certification task succeeds with probability at least $1 - \delta$.
\end{lemma}

\begin{proof}
  Since $\Truth\subseteq \Workers$ and $x_i$'s are positive numbers, the inequality 
$\sum_{i \in \Workers} x_i \ge (1-\eps) \sum_{i \in \Truth} x_i$ holds trivially. If the inequality 
$\sum_{i \in \Workers} x_i \le \frac 1 {1-\eps} \sum_{i \in \Truth} x_i$ does not hold, we can bound the probability that all 
of the $k$ verifications fail to find an invalid record, as follows:

The probability that a single verification fails to find an invalid record is 
$\sum_{i \in \Truth} p_i = \frac {\sum_{i \in \Truth} x_i} {\sum_{i \in \Workers} x_i} < 1 - \eps.$

Therefore the probability that all $k$ verifications fail is at most $(1-\eps)^k$. Setting 
$k = \Theta(\frac{1}{\eps}\log (1/\delta))$, we guarantee that an invalid record is found with probability at least $1-\delta$.
\end{proof}

\subsection{Functions given by Linear Programs}
\label{ssubs:coverLP}

  We now extend the previous results for the sum function to more general objective functions that can be represented 
as linear programs. We first consider the special case of packing and covering LPs while later we present a result for
general linear programs.

$$
  \begin{array}{cc}
 	
\textbf{Packing LP} &
\textbf{Covering LP} \\
  
  \begin{array}{llll}
    \max_y  & \displaystyle\sum\limits_{i \in \Workers} c_{i} & y_{i} & \\
    \text{s.t.} & \displaystyle\sum\limits_{i \in \Workers} a_{ij} & y_{i} \le b_{j},  &j = 1, \dots, m \\
    & & y_i \ge 0, & i \in \Workers
  \end{array}
  
  &
  \begin{array}{llll}
      \min_x  & \displaystyle\sum\limits_{j = 1}^m b_{j} & x_{j} & \\
      \text{s.t.} & \displaystyle\sum\limits_{j = 1}^m a_{ij} & x_{j} \ge c_{i},  & i \in \Workers \\
      & & x_j \ge 0, & j = 1, \dots, m
    \end{array}
  \end{array}
$$

  Packing and Covering LP's are parameterized by the \emph{non-negative} parameters $a_{ij}, b_j, c_i$. We assume that
each record $i$ contains all parameters under his control, i.e. the value $c_i$ and $a_{ij}$ for all $j$, while the 
parameters $b_j$ are accurately known in advance.

  Packing LPs capture settings where several resources (each available in a quantity $b_j$) are to be divided among 
a set of agents in the system and agents report how much of each resource they need (given by $a_{ij}$) and how much
value they can generate if they are given the resources they ask for (given by $c_{ij}$). Our goal is to compute an
efficient allocation to agents that maximizes the total value generated. For the certification task, we want to 
certify that the total value generated by the true agents in an optimal allocation is close to the value computed 
under the possibly incorrect reports. We show that efficient certification schemes exist by extending the 
certification scheme presented for the sum function:

\begin{lemma} \label{lem:packingLP}
    Let $a_{ij}, c_i  \ge 0$ be values contained in the records $\Workers$ and $y^*$ be the optimal solution to the 
  packing LP. Consider the probability distribution $p_i = \frac {y^*_i c_i} {\sum_j y^*_j z_j}$ which assigns records 
  a probability proportional to their computed value $y^*_i c_i$. Verifying $k=\Theta(\frac{1}{\eps}\log (1/\delta))$ 
  records sampled independently from $p$, guarantees that the certification task for the packing LP succeeds with 
  probability at least $1 - \delta$.
\end{lemma}

  To see why this lemma holds, notice that the value $\sum_{i \in \Workers} c_i y^*_i$ computed using all the records 
$\Workers$ is higher than the value $\sum_{i \in \Truth} c_i \bar y_i$ computed using only the valid records $\Truth$. 
Moreover, if $\sum_{i \in \Truth} c_i y^*_i \ge (1-\eps) \sum_{i \in \Workers} c_i y^*_i$, then it must be that
$\sum_{i \in \Truth} c_i \bar y_i \ge (1-\eps) \sum_{i \in \Workers} c_i y^*_i$ as well, since setting $y_i = y^*_i$ 
for $i\in \Truth$ and $y_i=0$ otherwise is a feasible solution to the packing LP under the valid records. Finally, if  
$\sum_{i \in \Truth} c_i y^*_i < (1-\eps) \sum_{i \in \Workers} c_i y^*_i$, it means that invalid records contribute 
more than an $\eps$ fraction of the total value and thus an invalid record can be easily found as in the previous case
of the sum function.

  Covering LPs naturally capture various settings with public goods where the designer wants to introduce new goods to
satisfy all the demands coming from the records of our data set, but minimizing the total cost at the same time. In 
facility location problems, the designer wants to open facilities so that all a set of agents have access to at least
one facility and the agents report which locations are accessible to them.

  Certification schemes for covering LPs are less direct than previous examples, but can be easily obtained through LP 
duality. As the dual of a covering LP is a packing LP which has the \emph{exact same value}, we can use the certification 
scheme of Lemma~\ref{lem:packingLP} to certify that value. We directly get the following:

\begin{lemma} \label{lem:coveringLp}
    Let $a_{ij}, c_i \ge 0$ be values in the records $\Workers$. Verifying $k=\Theta(\frac{1}{\eps}\log (1/\delta))$ records 
  sampled independently according to a distribution $p$ given by the solution to the dual packing LP, guarantees that the 
  certification task for the covering LP succeeds with probability at least $1 - \delta$.
\end{lemma}

  General LPs can be written in the form of a packing or a covering LP but have arbitrary (possibly negative) parameters 
$a_{ij}, b_j, c_i$. The value of such LPs is harder to certify in general as a lot more verifications than before might be 
needed. However, we can show that $m$ verifications suffice to certify their value exactly.

\begin{lemma} \label{lem:generalLp}
    Let $a_{ij}, c_i$ be (possibly negative) values contained in the records $\Workers$. The certification complexity for 
  general LPs (written in the form of packing or covering LPs above) is at most $m$.
\end{lemma}

  To see why this is true, notice that in the covering LP formulation, the optimal value is given by at most $m$ tight 
constraints as there are only $m$ variables. Verifying the $m$ records relevant to those constraints guarantees that the 
optimal value of the LP under only the value of the records is equal to the computed one. This is because only those $m$ 
constraints determine the optimal value and even if every other constraint $i$ was dropped (i.e. because $i \notin \Truth$)
the value would remain the same. The result also holds for general LPs under the packing LP formulation by LP duality.
  
  \section{Certification Schemes for $\vec{w}$-Lipschitz Functions and Applications} \label{sec:instOpt}

    In this section, we present a unified way of finding \textit{almost-optimal certification schemes}. For a given
a function $f$, a desired approximation parameter $\eps$ and an instance $\vec{x}_{\Workers}$, we want to compute the
``instance-optimal'' number of verifications in order to certify that
$f(\vec{x}_{\Workers}) \in \left[ 1 - \eps, \frac{1}{1 - \eps} \right] f(\xt)$ with probability of failure at most $1/3$.
The first result of this section is a structural result. We show that even though optimal schemes may be arbitrarily
complex, there are simpler schemes, that verify records independently, which are almost-optimal.
\smallskip

To show this we define for every set $S\subseteq \Workers$ the probability $p_S$ that the instance-optimal certification scheme
$\mathcal{C}^*$ verifies at least one record in $S$, i.e. $p_S = \Prob\left(\bigcup_{i \in S} \text{ \{$\mathcal{C}^*$ verifies
record $i$\}} \right)$. For such an event, we say that the \emph{certification scheme verifies $S$} and for simplicity we denote
$p_i$, the probability that $\mathcal{C}^*$ verifies record $i$, i.e. $p_i = p_{\{i\}}$.

    For the instance $\vec{x}_\Workers$, the set of invalid records could be any $S \subseteq \Workers$. For the certification
scheme to work with failure probability at most $2/3$, we must have that $p_{S} \ge 2/3$ for any subset $S \subseteq \Workers$
such that $f(\xw)/f(\vec{x}_{\Workers\setminus S}) \notin [1 - \eps, 1/(1 - \eps)]$. If this doesn't hold for some $S$, an
adversary could choose the set of invalid records to be $S$ and the certification scheme $\mathcal{C}^*$ would fail with
probability more than $1/3$. Moreover, even though the  optimal certification scheme $\mathcal{C}^*$ may verify records in a very
correlated way, we have that  $\sum_{i \in  S} p_i \ge p_{ S} \ge 2/3$ from a simple union bound. Therefore, the certification
scheme $\mathcal{C}^*$ must satisfy the following set of necessary conditions:
  \[ \sum_{i \in  S} p_i \ge 2/3 ~~~\forall S \subseteq \Workers \text{~ such that ~} \frac{f(\xw)}{f(\vec{x}_{\Workers\setminus S})} \notin \left[ 1 - \eps, \frac{1}{1 - \eps} \right] \]

  \noindent By linearity of expectation, the expected total number of verifications that $\mathcal{C}^*$ performs is,
  \[ \Exp[ \text{total number of verifications} ] = \Exp \left[ \sum_{i \in \Workers} \vec{1}\text{\{$\mathcal{C}^*$ verifies record $i$\}} \right] = \sum_{i \in \Workers} p_i \] 

\noindent The above imply that the value of the following linear program is a lower bound on the total number of verifications
needed by the optimal scheme $\mathcal{C}^*$ for this specific instance $\vec{x}_\Workers$.
  \begin{equation}
    \label{eq:neccLP1}
    \begin{array}{ll@{}ll}
      \min  & \displaystyle\sum\limits_{i \in \Workers} & ~p_i & \\
      \text{s.t.} & \displaystyle\sum\limits_{i \in S} & ~p_i \ge 2/3, & \forall S \subseteq \Workers, \frac{f(\xw)}{f(\vec{x}_{\Workers\setminus S})} \notin \left[ 1 - \eps, \frac{1}{1 - \eps} \right] \\
                  & 0 \le & ~p_i \le 1, & \forall i \in \Workers
    \end{array}
  \end{equation}

  \noindent Notice that the solutions to LP~\eqref{eq:neccLP1}, do not directly correspond to certification schemes with success
probability $2/3$. However, as we show, any solution to LP \eqref{eq:neccLP1} can be converted to a certification scheme with
number of verifications at most twice as many as the optimal value of LP \eqref{eq:neccLP1} and success probability $2/3$. Since
the optimal value of LP \eqref{eq:neccLP1} lower bounds the instance optimal number of verifications, our derived certification
scheme will be a 2-approximation to the instance optimal scheme.

\begin{definition}
    For a solution $\bar{p}$ of LP~\eqref{eq:neccLP1}, we define the certification scheme $\mathcal{C}_{\bar{p}}$ that verifies
  each record $i$ independently with probability $q_i=\min\{2\bar{p}_i,1\}$.
\end{definition}

  It is clear that the certification scheme $\mathcal{C}_{\bar{p}}$ uses in expectation at most twice as many verifications as
the optimal value of LP~\eqref{eq:neccLP1} and the instance optimal scheme. We now show that it also achieves, success probability
of $2/3$ as required.
\medskip

\noindent Assume that the subset of valid records is $\Truth = \Workers \setminus S$. The probability that the scheme
$\mathcal{C}_{\bar{p}}$ does not verify anyone in the set $S = \{s_1, \dots, s_m\}$ is
\[ \Prob(\text{$\mathcal{C}_{\bar{p}}$ doesn't verify $\Truth$}) = \Prob( (\text{$\mathcal{C}_{\bar{p}}$ doesn't verify $s_1$}) \wedge \dots \wedge (\text{$\mathcal{C}_{\bar{p}}$ doesn't verify $s_m$})) =\prod_{s \in S} (1 - q_s) \]

\noindent Since $\bar{p}$ is a feasible solution to LP~\eqref{eq:neccLP1}, the probability that some record from $S$ is verified is

\[ \Prob(\text{$\mathcal{C}_{\bar{p}}$ verifies $S$}) = 1 - \prod_{s \in S} (1 - q_s) \ge 1 - \exp\left( - 2\sum_{i \in S} \bar{p}_s \right) \ge 1 - \exp\left( - 4/3 \right)\ge 2/3. \]
\noindent This means that our certification scheme succeeds with probability $2/3$ using at most twice the optimal number of
verifications in expectation. We can amplify the probability of $2/3$, making it arbitrarily close to one by repeating the
certification scheme. Since the repetitions are independent and each of them fails with probability at most $1/3$, after $r$
repetitions the total probability of failure is $3^{-r}$. Repeating $r = \log(1 / \delta)$ times, guarantees that for any
subset $S$, the probability that it will be verified is at least $1-\delta$. This result is summarized in the following theorem.

\begin{theorem} \label{thm:optInstV}
    For any given function $f : \Domain^* \to \reals$ and any set of valus if records $\xw$, a solution to LP~\eqref{eq:neccLP1} corresponds to a
  certification scheme that verifies records of the data set independently using at most twice as many verifications as the optimal scheme for
  this instance and succeeds with probability $2/3$. Repeating the scheme $\log(1 / \delta)$ times increases the success probability to
  $1-\delta$.
\end{theorem}

\paragraph{Remark} We note that the LP~\eqref{eq:neccLP1} has exponentially many constraints and it may be computationally
intractable to solve depending on the function. It is very useful though as a tool to uncover the structure of approximately
optimal certification schemes. For example, Theorem \ref{thm:optInstV} implies that even though optimal schemes may be arbitrarily
complex, there are simpler schemes, that verify records independently, which are almost-optimal.
\medskip

In the following section, we derive a general methodology to obtain solutions to LP~\eqref{eq:neccLP1} for the very general class of
$\vec{w}$-Lipschitz functions.

\subsection{Certification Schemes for $\mathbf{w}$-Lipschitz Functions}

  In this section we show how we can use Theorem \ref{thm:optInstV} to get sufficient smoothness conditions on the function $f$
that can be used to provide certification schemes with small number of verifications.

  For any record $i \in \Workers$ we define $w_i$ to be the \textit{weight} of the record $i$. The weight of record $i$ will be
the quantity that will determine the probability that we will verify record $i$ according to the verification scheme that we want
to define. We state now the property that we want $f$ to satisfy in order to find a good verification scheme.

\begin{definition} \label{def:wcont}
  We say that a function $f : \Domain^* \to \reals$ is $\vec{w}$-\emph{Lipschitz}, with $\vec{w} \in \reals_+^n$, if for any
  $S \subseteq \Workers$
  \[ |f(\xw) - f(\vec{x}_{\Workers\setminus S})| \le \sum_{i \in S} w_i \]
\end{definition}

\noindent For any function that satisfies this Lipschitz property we can get a good verification scheme that depends on the weight
vector $\vec{w}$.

\begin{theorem} \label{thm:wcont}
    For any non-negative $\vec{w}$-Lipschitz function $f : \Domain^* \to \reals_+$, set of records $\Workers$ with value $\xw$, and
  real numbers $\eps, \delta > 0$,
  \noindent there exists a certification scheme that uses at most
  $\frac{4\sum_{i \in \Workers} w_i}{3 f(\xw) \eps} \log(1 / \delta) $
  verifications and has probability of success at least $1 - \delta$.
\end{theorem}

  In Appendix \ref{sec:app:proofOfLipschitz} we present the proof of Theorem \ref{thm:wcont}. Also, in Appendix \ref{ssec:tspIOpt} and \ref{ssec:steinerIOpt}, we present two applications of Theorem \ref{thm:wcont} to get
certification schemes for the \emph{Traveling Salesman (TSP)} and the \emph{Steiner Tree} problems. In both applications, we show that
the optimal solution is $\vec{w}$-Lipschitz with $\frac{\sum_{i \in \Workers} w_i}{f(\xw)} \le 2$. Hence, the total number of
verifications by Theorem \ref{thm:wcont} is $O((1/\eps) \log(1/\delta))$.
 
  \clearpage
  \section{Weak Correction Model}
\label{sec:wcorr}

  We show how starting from a certification scheme, we can obtain a weak-correction scheme with the same verification complexity
(up to constants).

\begin{theorem}\label{thm:reduction}
    Suppose that there exists a certification scheme for a function $f$ that uses $q(n,\eps)$ verifications and fails with probability $1/3$. Then,
  there exists a weak-correction scheme with verification complexity $O(q(n,\eps) \log(1/\delta))$ that outputs an accurate estimate of the function
  $f$ and fails with  probability $\delta$.
\end{theorem}

  Theorem \ref{thm:reduction} shows that the certification task we defined in section \ref{sec:certification} is already strong enough to perform this
seemingly more challenging task. Intuitively, this is because we can run many rounds of certification until we have enough confidence that we have an
accurate result while we remove from the dataset any invalid record we might find during these rounds. Indeed, a simple way to make the conversion
is to start from a certification scheme with error probability 1/3 and reduce its probability of error to $\delta/n$, by repeating it $\log(n/\delta)$
times. Then use this scheme repeatedly until no more invalid records are detected. By a union bound, the probability of error is at most $\delta$ since
the process takes at most $n$ steps.
  Theorem \ref{thm:reduction} shows a stronger result than the above result showing that the logarithmic dependence on the number of records can be avoided
if the stopping time is more carefully chosen.

We provide here a simple analysis when the function $f$ is \emph{increasing} with record values and defer the
full proof of Theorem~\ref{thm:reduction} to Appendix~\ref{app:weak}.

\begin{proof}[Proof of Theorem \ref{thm:reduction} for increasing functions]
  Our weak correction scheme works by repeating the certification process enough times so that the number of times it failed is less than the number of
times it succeeded. In particular, we model this procedure as a random walk on the integers starting from point C and ending once it reaches 0. We move
to the right whenever the round of verifications (i.e an execution of the certification scheme) reveals some invalid record, and we move to the left
otherwise.

  The random walk is guaranteed to return to the origin eventually since if all invalid records are removed the certification scheme will not be able to
find any additional invalid record. The only case that the weak correction scheme fails is if it returns early without removing enough invalid records
having a value larger than $f(\xt)/(1-\eps)$. In such a case, at all points of the random walk the estimate was always larger than $f(\xt)/(1-\eps)$ which
means that the random walk was biased with probability at least $2/3$ to the right. The probability that such a biased random walk reaches the origin is
at most $\left( \frac {1/3} {2/3} \right)^C = 2^{-C}$. Setting $C = \log(1/\delta)$ times guarantees a probability of error $\delta$.
The number of verifications performed if $k$ invalid records are found is $(C + 2 k) q(n,\eps)$, thus the total verification complexity is
$O(q(n,\eps) \log(1/\delta))$.
\end{proof}

  \section{Strong Correction Model} \label{sec:scorr}

  In section we present cases where it is possible to obtain much more efficient correction schemes that do not involve
repetition of the certification scheme as in the case of weak correction. We first obtain a strong correction scheme for the 
sum function which also implies strong correction schemes for other functions such as average. However as we show there are 
simple functions, e.g. the composition of the max and the sum function, which do not admit strong correction schemes but for 
which good weak correction schemes exist. Finally we show that despite the previous lower bound there is a rich class of
problems for which a strong correction scheme exist. We do so by proving a connection of strong correction schemes with
algorithms that use conditional sampling.

\subsection{Computing the Sum of Values of Records}

  We use the same formulation as in Section \ref{ssec:certificationSum} and we get the following result.
  
\begin{lemma} \label{lem:sumC}
    Let $x_1, x_2, \dots, x_n \ge 0$ be the values of the records $\Workers$ and $f(\xw) = \sum_{i \in \Workers} x_i$.
  Consider the probability distribution $p_i = \frac {x_i} {\sum_j x_j}$ which assigns each record $i$ a probability 
  proportional to their value $x_i$. If we sample $M$ times independently from $p$ and verifying records until 
  $k = \Theta\left(\frac{1}{\eps^2}\log (1/\delta)\right)$ valid records found, then the estimator 
  $\hat{s} = \frac{k}{M} \sum_{i \in \Workers} x_i$ is in the range 
  $\left[ 1 - \eps, \frac{1}{1 - \eps} \right] \cdot \sum_{i \in \Truth} x_i$ with probability at least $1 - \delta$. 
  This gives a strong correction scheme with probability of success at least $1 - \delta$.
\end{lemma}

  A detailed proof of Lemma \ref{lem:sumC} can be found in Appendix \ref{sec:app:scorr} were we prove also that $\Omega(1/\eps^2)$ verifications 
are nececssary.

\subsection{Lower Bound for the Maximum of Sums Function}

  In this section we show that no efficient strong correction scheme exists for the composition of the max and the sum function. 
More precisely we assume we have a partition $\mathcal{J} = \{\Workers_1, \dots, \Workers_{\ell}\}$ of the set $\Workers$ and we 
want a strong correction scheme for the function $f(\xw) = \max_{A \in \mathcal{J}} \sum_{i \in A} x_i$. We show that any strong 
correction scheme for $f$ that achieves constant approximation has to verify at least a constant fraction of the records.

\begin{lemma} \label{lem:maxOfSumC}
    Let $c \in \reals_+$ then there exists a partition $\mathcal{J} = \{\Workers_1, \dots, \Workers_{\ell}\}$ of the set of records 
  $\Workers$ and a vector $\xw \in \reals^n$ such that any strong correction scheme for the function 
  $f(\xw) = \max_{A \in \mathcal{J}} \sum_{i \in A} x_i$, that returns an estimator $\hat{s}$ such that 
  $\hat{s} \in \left[\frac{1}{c}, c\right] \cdot f(\xt)$ with probability at least $3/4$, has to verify at least 
  $\abs{\Workers} / 4 c^2$ records.
\end{lemma}

  A detailed proof of Lemma \ref{lem:maxOfSumC} can be found in Appendix \ref{sec:app:scorr}.

\subsection{From Algorithms using Conditional Sampling to Strong Correction Schemes}

  The design of a strong correction scheme is sometimes a very hard task since the guarantee is very strong. Our main theorem 
in this section shows that there is a nice correspondence of a strong correction scheme with 
\textit{sublinear algorithms using conditional sampling}, a model that has been appeared recently in \cite{GouleakisTZ2017}.
The applications of this framework involve problems expressed in the $d$-dimensional Euclidean space. We state here the main
theorem for this section and we defer the full discussion for Appendix \ref{sec:app:applicationsStrongCorrection} where wee present also
some important applications of this theorem.

\begin{theorem} \label{thm:sCorrection}
    An algorithm that uses $k$ conditional samples to compute a function $f$ can produce a strong correction scheme with 
  verification cost $k$.
\end{theorem} 
  \section*{Acknowledgements}
  The authors were supported by NSF CCF-1551875, CCF-1617730, CCF-1733808, and
  IIS-1741137.

  \bibliographystyle{alpha}
  \bibliography{ZAMPETAKIS18}
  \appendix

  \section{Applications of Theorem \ref{thm:wcont}}

\subsection{Optimal Travelling Salesman Tour} \label{ssec:tspIOpt}

In this section we examine the \textit{metric travelling salesman problem} where we are given $n$ points (each provided by one record in $\Workers$) in a metric space
$\mathcal{X}$ and we wish to find the length of the minimum cycle going through each point in the set $\Truth\subseteq \Workers$ of correct answers. As usual we let
$\vec{x}_\Workers$ be the input vector with record values whose coordinates are points in the metric space $\mathcal{X}$. Our goal is to find a certification
scheme for this metric travelling salesman problem. That is, the algorithm should either output a sufficiently accurate value (according to \ref{eq:verH0}) for the
minimum weight cycle going through the points in $\vec{x}_\Truth$  or find a invalid record \footnote{Note that throughout this paper we don't consider the
computational complexity of the problems, since we are more interested in the number of verifications needed. Besides that in the case of Euclidean TSP we could
use the $(1 + \eps)$-approximation algorithm that we know in order to get similar results and avoid $\NP$-completeness.}. The following lemma combined with Theorem \ref{thm:wcont} give us
the desired result.

\begin{lemma}\label{lm:tsp}
    Let $f : \Domain^* \to \reals$ be the function mapping a set of points in a metric space $\mathcal{X}$ to their minimum TSP tour and let $v_1v_2\dots v_n$ be the
  minimum TSP tour. Also, let $\vec{w} \in \reals_+^n=(w_1,\dots,w_n)$, where
  $w_i=d(v_{i-1},v_i)+d(v_i,v_{i+1})$ and the second indices are mod $n$. Then, $f$ is $\vec{w}$-continuous.
\end{lemma}

\begin{proof}
According to definition \ref{def:wcont}, we need to show that for any $S\subseteq\Workers$:

    \begin{equation}\label{tsp:ineq} f(\xw) \le f(\vec{x}_{\Workers\setminus S}) + \sum_{i \in S} w_i
    \end{equation}
To see why this inequality is satisfied, 
let $T_R$ be the minimum TSP tour going through the points in $R=\Workers\setminus S$ and $T_\Workers=v_1v_2\dots v_n$ be the minimum TSP tour that goes through all the points in the
set $\Workers\supseteq R$. Now let $j_1<j_2\dots <j_r$ be the indices at which the points of the set $R$ appear in this TSP tour.  Consider two consecutive points
$v_{j_k},v_{j_{k+1}}$ in this sequence and let $P_k=\{v_{j_k+1},v_{j_k+2},\dots, v_{j_{k+1}-1}\}$ be the set of consecutive points in the tour $T_\Workers$ between
$v_{j_k}$ and $v_{j_{k+1}}$. Clearly, $\forall k: P_k\subseteq  S$ and therefore the weights of those points appear in the sum that is in the rhs of
equation \eqref{tsp:ineq}. Now consider the two paths $p_{1,k}=v_{j_k},v_{j_k+1},\dots, v_{j_{k+1}-1}$ and $p_{2,k}=v_{j_k+1},v_{j_k+1},\dots, v_{j_{k+1}}$ which are
both part of $T_\Workers$. We have that:
\[
\sum_{i\in P_k} w_i = d(v_{j_k},v_{j_k+1})+d(v_{j_{k+1}-1},v_{j_{k+1}})+2\cdot \sum_{i=j_k+1}^{j_{k+1}-2} d(v_i,v_{i+1})=l(p_{1,k})+l(p_{2,k})
\]
where $l(\cdot)$ denotes the length of a path. 
We now consider the walk that goes through all the vertices in $\Workers$ and has the following two properties:
\begin{itemize}
  \item It respects the order in which the vertices in $R$ are visited by $T_R$
  \item Between any two consecutive such vertices, it follows whichever path among $p_{1,k}$ and $p_{2,k}$ has smaller length in the forward and then backwards
direction.
\end{itemize}

We know that $f(\xw)$ smaller or equal to the walk we have just defined, since the walk goes through all the given points and even repeats the points in
$R$ \footnote{Since we are working on a metric space, skipping points in the order that we visit them can only decrease the cost.}. Thus,

\begin{align*}
 f(\xw) &\le f(\vec{x}_{\Workers\setminus S}) + \sum_{k=1}^s 2\cdot \min\{(d(v_{j_k},v_{j_k+1}),d(v_{j_{k+1}-1})\}+ 2\cdot \sum_{i=j_k+1}^{j_{k+1}-2} d(v_i,v_{i+1}) \\
 &\le f(\vec{x}_{\Workers\setminus S}) + \sum_{k=1}^s\sum_{i\in P_k}w_i\\
 &\le f(\vec{x}_{\Workers\setminus S}) + \sum_{i \in  S} w_i
\end{align*}
\end{proof}

Using lemma \ref{lm:tsp} and theorem \ref{thm:wcont}, we get the following corollary:

\begin{corollary}
    Let $f : \Domain^* \to \reals$ be the function mapping a set of points in a metric space $\mathcal{X}$ to their minimum TSP tour. Then, there exists a
  verification scheme that uses at most $O(\frac{1}{\eps}\log(\frac{1}{\delta}))$ verifications per correction.
\end{corollary}

\begin{proof}
This is a straightforward application of lemma \ref{lm:tsp} and theorem \ref{thm:wcont} since $\sum_{i\in \Workers} w_i$ contains each of the edges in the optimal TSP tour $T_\Workers$ exactly twice. Thus,
\[
\sum_{i\in \Workers} w_i=2f(\vec{x}_\Workers)
\]
\end{proof}
 
\subsection{Steiner tree} \label{ssec:steinerIOpt}

In the classic Steiner tree problem, the input is a positively weighted graph $G=(V,E,w)$ and the set of vertices $V$ is partitioned into two disjoint sets $T$ and
$U$ such that $V=T\cup U$. Usually $T$ is called the set of \emph{terminal} nodes and $U$ the set of \emph{Steiner} nodes. The goal is to compute a connected
subgraph of $G$ that has the smallest possible weight and has a set of vertices $T\subseteq V^\prime \subseteq V$ that includes all \emph{terminal} nodes and any number of steiner nodes.

  Here, we are going to examine the Steiner tree problem in the following setting: We are given a fixed graph $G=(V,E)$ on $\vert V\vert$ vertices and we also have
$\vert \Workers\vert$ values from the set of records $\Workers$. Each record is a node from the set $V$ claiming that this node is in the set $T\subseteq V$ of
terminal nodes that need to be connected by the tree. However, the records might be invalid and the algorithm is allowed to do verifications on those records.
Let $\vec{x}_\Workers$ be the input vector whose coordinates are vertices claimed to be in the set $T$ of terminal nodes. Similarly, let $\vec{x}_A$ be a vector
containing only a subset $A\subseteq \Workers$ of those vertices. Our goal is again to be able to either output a sufficiently accurate answer for the cost of the
optimal Steiner tree of find an invalid record.

As in the previous section we are going to use theorem \ref{thm:wcont} to achieve this. The conditions of theorem \ref{thm:wcont} are satisfied in this case due to
the following lemma:

\begin{lemma}\label{lm:steiner}
    Let $G=V,E$ be a graph and $f_G : V^* \to \reals$ be the function mapping a set of vertices $T\subseteq V$ to the minimum cost of a steiner tree connecting the
  vertices in $T$.  Then, there exists a vector $\vec{w}\in \reals_+^n$ such that $f$ is $\vec{w}$-continuous and also
  $\sum_{i\in \Workers}w_i = O(f_G(\vec{x}_\Workers))$.
\end{lemma}

\begin{proof}
We need to show that there exists a vector $\vec{w}\in \reals_+^n=(w_1,\dots,w_n)$, such that for any $S\subseteq\Workers$, the following inequality holds:

 \begin{equation}\label{steiner:ineq} f(\xw) \le f(\vec{x}_{\Workers\setminus S}) + \sum_{i \in S} w_i
 \end{equation}

We start by introducing some notation. Let $t$ be a tree subgraph of $G$. We denote by $H_t$ the Eulerian graph that results when we double each edge in $t$. Also,
let $t_A$ denote the optimal Steiner tree for the set $A\subseteq V$ of terminal nodes. Thus, $\forall A: f(\vec{x}_A)=cost(t_A)$.

Now let $t_R$ be the optimal Steiner tree for some set $R=\Workers\setminus S\subseteq V$ of terminal nodes. In order to show equation \eqref{steiner:ineq}, it suffices to show that
there exists a tree $t$ and a vector $\vec{w}\in \reals_+^n$, such that $t$ is a valid Steiner tree for the set $\Workers$ of terminal nodes and its cost is:
$cost(t)\le cost(t_R)+\sum_{i\in S}w_i$.

In other words, we would like to find a weight vector $\vec{w}\in \reals_+^n$, such that starting from the Steiner tree $t_R$ and using the weight assigned to the
set $ S=\Workers\setminus R$ as budget, we are able to construct a Steiner tree the \emph{covers} the set $\Workers$. To keep the number of verifications low, we also
require this vector to be such that $\sum_{i\in \Workers}w_i =O(f_G(\vec{x}_\Workers))$.

Now fix a specific Euler tour (i.e an ordering of the nodes) $U_\Workers$ for the graph $H_{t_\Workers}$ and also fix an Euler tour  $U_R$ for the graph $H_{t_R}$.
Note that the cost of each Euler tour is exactly twice the cost of the corresponding Steiner tree (e.g $cost(U_R)=2cost(t_R)$ where $cost(\cdot)$ denotes the sum of
weights of all edges in the Euler tour or the tree).

We define each weight $w_i$ to be the length of the path from the predecessor to the successor of node $i$ in the ordering $U_\Workers$.

Our goal is to find a new Euler tour which directly corresponds to a valid Steiner tree \footnote{That is, the traversing each edge of that tree twice and in
opposite directions.} for the set $\Workers$ and is within our budget $\sum_{i\in S}w_i$.

Now let $U_\Workers=v_1v_2\dots v_n$ be the ordering in which the terminal nodes are visited in the Euler tour of $H_{t_\Workers}$ and $j_1<j_2\dots <j_r$ be the
indices at which the points of the set $R=\Workers\setminus S$ appear in this Euler tour.  Consider two consecutive points $v_{j_k},v_{j_{k+1}}$ in this sequence and let
$P_k=\{v_{j_k+1},v_{j_k+2},\dots, v_{j_{k+1}-1}\}\subseteq  S$ be the set of consecutive points in the Euler tour $U_\Workers$ between $v_{j_k}$ and
$v_{j_{k+1}}$. Note that the sets $P_k$ are mutually disjoint and therefore: $\sum_{k=1}^r \sum_{i\in P_k}w_i \le \sum_{i\in  S}w_i$. Also,
$\sum_{i\in P_k}w_i$ is enough budget to add the set of nodes $P_k$ in the ordering $U_R$ between $v_{j_k}$ and $v_{j_{k+1}}$. \footnote{To be more precise here, we
need an argument similar to the two paths argument in the proof of lemma \ref{lm:tsp}.}  By repeating this for all $k\in [r]$, we get the desired Steiner tree $t$
that \emph{covers} all nodes in $\Workers$ and is such that:

\[
2\cdot cost(t_\Workers)\le 2\cdot cost(t) \le 2\cdot cost(t_R) +   \sum_{i\in  S}w_i \Rightarrow
\]

\[
 cost(t_\Workers)\le  cost(t_R) +   \sum_{i\in  S}\frac{w_i}{2} \Leftrightarrow
\]

\[
 f(\vec{x}_\Workers)\le  f(\vec{x}_{\Workers\setminus S}) +   \sum_{i\in S}w_i^\prime
\]where $w_i^\prime=\frac{w_i}{2}$.

Thus, $f$ is $\frac{\vec{w}}{2}$-continuous and also $\sum_{i\in \Workers}w_i^\prime = \frac{1}{2}\cdot 2\cdot cost(U_\Workers)=2\cdot f(\vec{x}_\Workers)$.
\end{proof}

The following corollary is a direct application of  lemma \ref{lm:steiner} and theorem \ref{thm:wcont}:

\begin{corollary}
Let $G=V,E$ be a graph and $f_G : V^* \to \reals$ be the function mapping a set of vertices $T\subseteq V$ to the minimum cost of a steiner tree connecting the
vertices in $T$. Then, there exists a verification scheme that uses at most $O(\frac{1}{\eps}\log(\frac{1}{\delta}))$ verifications per correction.
\end{corollary}
 
\section{Proof of Theorem~\ref{thm:reduction}}\label{app:weak}

We will now remove the assumption we used earlier about the function $f$ being increasing in order to design a weak correction scheme.
We are going to use the same random walk based correction scheme as in Section~\ref{sec:wcorr} that starts at $C$ and ends at 0. However, instead of outputting the result of the function $f$ on the final subset of records (after all deletions), we will consider every possible intermediate subset of records during the random walk as a candidate for producing an $(1+\varepsilon)$-approximate solution. Note that, at each step $i$ of the random walk, we run a certification scheme on some set $S_i\subseteq \Workers$. We define a subset $S\subseteq\Workers$ to be ``bad'' if $\frac{f(\vec{x}_{\Truth})}{f(\vec{x}_S)} \not\in \left[ 1 - \eps, \frac{1}{1 - \eps} \right]$ and to be ``good'' otherwise.

By the definition of the certification scheme, if the set $S$ is ``bad'', then an invalid record is found with probability at least $2/3$, in which case the random walk moves to the right. Otherwise, we do not have any guarantee on how the random walk will behave.

However, if at all steps the probability of finding an invalid record is more than $3/5$, then the probability that the random walk reaches 0 is less than $(\frac {2/5} {3/5})^C = (\frac 2 3)^C < \delta/2$ for $C=O(\log(1/\delta))$. Thus given that we returned, with high probability, there must be some set $S_i$ for which the correction scheme accepts with probability more than $2/5$. Note that, this can only be true if the set $S_i$ is good since $2/5 > 1/3$.

At this point, given a list of these subsets, our goal is to find a ``good'' subset for which the certification scheme accepts with probability more than $1/3$. We know that a ``good'' subset exists for which the acceptance probability is more than $2/5$. We view the certification process for a subset $S$ as sampling from a Bernoulli random variable. We say that a set $S$ has probability $p$ if the certification process on the set $S$ does not find an invalid record with probability $p$.

Let $Test(S,\gamma)$ be a test that accepts if the probability of a set $S$ is more than $2/5$ (call such a set ``very good'') and rejects if it is less than $1/3$. Such a test fails with probability $\gamma$ requiring $O(\log(1/\gamma))$ samples.

The main idea behing this algorithm is to iteratively run $Test(S,\gamma)$ for all candidate subsets $S$ with varying error probabilities $\gamma$ to throw out the failing ones until a significant fraction of the subsets in our pool is ``good''.  When this happens, we pick a subset at random and check if it is actually ``good'' by running $Test(S,\gamma)$ with small $\gamma$. We repeat this until we actually find a good subset and output the value on the function $f$ on that subset. To ensure that this will eventually happen, we choose parameters appropriately, so that a constant fraction of the ``bad'' subsets fail while the ``good'' subsets pass the certifications with high enough probability.

Let $K$ be the number of candidate subsets $S_i$. We have that $K$ is equal to the number of invalid records found during the random walk process.

Our algorithm proceeds in rounds until there are at most $K / \log K$ sets remaining. In the $t$-th round:
\begin{itemize}
  \item The algorithm runs $Test(S_i,10^{-t})$ for every set $S_i$ and discards all sets that fail.
  \item If the number of remaining sets is did not drop by a factor of 2 the algorithm stops and returns a set $S_i$ uniformly at random from the remaining sets.
\end{itemize}
If the algorithm has not returned after $\log \log K$ steps, then it runs $Test(S_i,1/K^2)$ for every remaining set $S_i$ and returns one that passes the test.

The proposed algorithm returns a ``good'' set with probability more than $3/5 - o(1)$.
First, notice that a ``very good'' set will be discarded in the first $\log \log K$ rounds with probability at most $\sum_t 10^{-t} \le \frac {1/10} {1 - 1/10} = 1/9$.
Hence, if the algorithm did not return after $\log \log K$ rounds, the last step returns a ``good'' set with high probability.

Now, suppose the algorithm returns at some round $t$. Let $K_{t-1}$ be the total remaining sets before round $t$. The probability that the number $B_t$ of ``bad'' sets remaining after round $t$ to be more than $K_{t-1}/5$ is at most:
$$\Prob\left(B_t > K_{t-1}/5\right) \le \exp(-B_{t-1}/10) \le \exp(-K_{t-1}/50) \le \exp(-K/(50 \log K))$$
This is an exponentially small probability and by a union bound over all $\log \log K$ rounds it is still negligible.

Thus, assuming that $B_t < K_{t-1}/5$ and $K_t > K_{t-1}/2$, a set chosen uniformly at random is ``bad'' with probability $B_t/K_t \le 2/5$.

Therefore, a good set is chosen with probability at least $3/5 - o(1)$ and thus by repeating $O(\log(1/\delta))$ times and choosing the median of the values $f(\vec{x}_S)$, we have that with probability $1 - \delta/2$, $\frac{f(\vec{x}_{\Truth})}{f(\vec{x}_S)} \in \left[ 1 - \eps, \frac{1}{1 - \eps} \right]$.

The total number times the certification scheme is called is $O(\log(1/\delta) ) \sum_t O(2^{-t} K \log 10^{t}) = O( K \log(1/\delta) )$.

Thus, the verification complexity of the weak correction scheme is equal to $O(q(n,\eps) \log(1/\delta))$ and the Theorem follows.
 
\section{Missing Proof of Section \ref{sec:scorr}} \label{sec:app:scorr}

\begin{proof}[Proof of Lemma \ref{lem:sumC}]
  We let the random variable $M$ to be the total number of verifications until we found $k$ valid records and let $\mathcal{M}$ be the set of samples that we observed. Also we define 
$Z = \abs{\mathcal{M} \cap \mathcal{T}}/M = k/M$. We claim that
\[ \Prob\left(\frac{\sum_{i \in \Truth} x_i}{\sum_{i \in \Workers} x_i} \in \left[ 1 - \eps, \frac{1}{1 - \eps} \right] \cdot Z \right) \ge 1 - \delta \]

\noindent given that $\abs{\mathcal{M} \cap \mathcal{T}} \ge k$.

  Let $q = \frac{\sum_{i \in \Truth} x_i}{\sum_{i \in \Workers} x_i}$. For tha sake of contradiction let $Z > \frac{1}{1 - \eps} \frac{\sum_{i \in \Truth} x_i}{\sum_{i \in \Workers} x_i}$ then 
$M < (1 - \eps) k / q$. Hence the expected number of valid records if we draw $M$ samples according to the described distribution is at most $(1 - \eps) k$. But know using simple Chernoff 
bounds and the fact that $k \ge \frac{1}{\eps^2}\log (2/\delta)$ we get that with probability at most $\delta / 2$ the number of valid records found is at least $k$. 

  Similarly we can show that if $Z < (1 - \eps) \frac{\sum_{i \in \Truth} x_i}{\sum_{i \in \Workers} x_i}$ then with probability at most $\delta / 2$ the number of valid records found is at most $k$. Hence we have
\begin{align*}
  \Prob(\abs{\mathcal{M} \cap \mathcal{T}} = k) & = \Prob\left(\abs{\mathcal{M} \cap \mathcal{T}} = k \mid q \in \left[ 1 - \eps, \frac{1}{1 - \eps} \right] \cdot Z \right) \Prob\left( q \in \left[ 1 - \eps, \frac{1}{1 - \eps} \right] \cdot Z \right) \\
               & + \Prob\left(\abs{\mathcal{M} \cap \mathcal{T}} = k \mid q < (1 - \eps) Z \right) \Prob\left( q < (1 - \eps) Z \right) \\
               & + \Prob\left(\abs{\mathcal{M} \cap \mathcal{T}} = k \mid q > \frac{1}{1 - \eps} Z \right) \Prob\left( q > \frac{1}{1 - \eps} Z \right)
\end{align*}
\noindent but from the definition of the strong correction scheme $\Prob(\abs{\mathcal{M} \cap \mathcal{T}} = k) = 1$ and as we proved 
\[ \Prob\left(\abs{\mathcal{M} \cap \mathcal{T}} = k \mid q < (1 - \eps) Z \right) \le \delta / 2 ~~ \text{ and } \]
\[ \Prob\left(\abs{\mathcal{M} \cap \mathcal{T}} = k \mid q > \frac{1}{1 - \eps} Z \right) \le \delta / 2 \]
\noindent therefore
\begin{align*}
  \Prob(\abs{\mathcal{M} \cap \mathcal{T}} = k) = 1 & \le \Prob\left(\abs{\mathcal{M} \cap \mathcal{T}} = k \mid q \in \left[ 1 - \eps, \frac{1}{1 - \eps} \right] \cdot Z \right) \Prob\left( q \in \left[ 1 - \eps, \frac{1}{1 - \eps} \right] \cdot Z \right) + \delta \\
\end{align*}
\noindent which implies
\[ \Prob\left( q \in \left[ 1 - \eps, \frac{1}{1 - \eps} \right] \cdot Z \right) \ge 1 - \delta. \]
  This finally implies that our estimator is in the correct range
\[ \Prob\left( Z \cdot \sum_{i \in \Workers} x_i \in \left[ 1 - \eps, \frac{1}{1 - \eps} \right] \cdot \sum_{i \in \Truth} x_i \right) \ge 1 - \delta. \]

  To see that $\Theta\left(\frac{1}{\eps^2}\log (1/\delta)\right)$ are also necessary let $x_1 = \cdots = x_n = 1 / n$ and let 
$\abs{\Truth} = \abs{\Workers} q$ where $q = 1/2$. This instance is identical with estimating the
bias of a Bernoulli random variable with error at most $\eps$ and since all the $x_i$'s are equal we can assume without loss 
of generality that at each step we take a uniform sample from $\Workers$. But it is well known that for
estimating a Bernoulli random variable within $\eps$ with probability of failure at most $\delta$ we need at least 
$\Theta\left(\frac{1}{\eps^2}\log (1/\delta)\right)$ total samples. Half of those samples are expected to be correct samples
and hence the verification complexity for any strong correction scheme is also at least 
$\Theta\left(\frac{1}{\eps^2}\log (1/\delta)\right)$.
\end{proof}

\begin{proof}[Proof of Lemma \ref{lem:maxOfSumC}]
  We consider a partition $\mathcal{J}$ of $\Workers$ into $n/c^2$ sets of size $c^2$ each with $x_i = 1$ for all $i \in \Workers$. Let $S$ be a certification scheme that verifies less than $n / 4 c^2$ records. Then there 
exists a set $A \in \mathcal{J}$ such that
\[ \Prob( S \text{ verifies some } j \in A ) < 1/4. \]
\noindent We prove this by contradiction. Let $\Prob( C \text{ verifies some } j \in A ) \ge 1/4$ for all $A \in \mathcal{J}$. Then 
$\Exp[\text{verification by } C] \ge \sum_{A \in \mathcal{J}} \Prob( C \text{ verifies some } j \in A ) \ge n / 4 c^2$ and hence we have a contradiction on the assumption that $S$ verifies less than $n / 4 c^2$ records. Let $\hat{s}$
be the output estimator of $S$ then we have that
\begin{align*}
  \Prob\left(\hat{s} \in \left[\frac{1}{c}, c\right] \cdot f(\xt)\right) & = \Prob\left(\hat{s} \in \left[\frac{1}{c}, c\right] \cdot f(\xt) \mid  S \text{ verifies some } j \in A\right) \Prob\left( S \text{ verifies some } j \in A\right) \\
                                                                         & + \Prob\left(\hat{s} \in \left[\frac{1}{c}, c\right] \cdot f(\xt) \mid  S \text{ does not verify } A\right) \Prob\left( S \text{ does not verify } A\right) \\
                                                                         & < 1/4 + \Prob\left(\hat{s} \in \left[\frac{1}{c}, c\right] \cdot f(\xt) \mid  S \text{ does not verify } A\right)
\end{align*}
\noindent Now if we fix $Q \subseteq \nats$, we observe that the quantity $\Prob\left(\hat{s} \in Q \mid  S \text{ does not verify } A\right)$ does not depend on $\Truth \cap A$ since we are conditioning on the event that $S$
does not verify any record in $A$. Now let $j_B$ be an arbitrary record from the set $B \in \mathcal{J}$. We consider the following two possibilities for the set $\Truth$.
\[ \Truth_0 = \bigcup_{B \in \mathcal{J}, B \neq A} \{j_B\} \]
\[ \Truth_1 = \Truth_0 \cup A \]
\noindent We observe now that if $\Truth = \Truth_0$ then $f(\xt) = 1$ and if $\Truth = \Truth_1$ then $f(\xt) = c^2$. Now since $\hat{s}$ does not depend on $\Truth \cap A$ given that $S \text{ does not verify } A$ we have that we 
can change $\Truth$ between $\Truth_0$ and $\Truth_1$ without changing the quantity $\Prob\left(\hat{s} \in Q \mid  S \text{ does not verify } A\right)$. Now 
\begin{itemize}
  \item[-] if $\Prob\left(\hat{s} \in [1, c] \mid  S \text{ does not verify } A\right) < 1/2$ then we set $\Truth = \Truth_1$ and 
  \item[-] if $\Prob\left(\hat{s} \in [c(c - 1), c^2] \mid  S \text{ does not verify } A\right) < 1/2$ then we set $\Truth = \Truth_1$. 
\end{itemize}

\noindent Observe that one of the two cases has to be true. In any of these we get that 
\[ \Prob\left(\hat{s} \in \left[\frac{1}{c}, c\right] \cdot f(\xt) \mid  S \text{ does not verify } A\right) < 1/2. \]

\noindent Hence we get that  
\[ \Prob\left(\hat{s} \in \left[\frac{1}{c}, c\right] \cdot f(\xt)\right) < 1/4 + \Prob\left(\hat{s} \in \left[\frac{1}{c}, c\right] \cdot f(\xt) \mid  S \text{ does not verify } A\right) < 3/4 \]
\noindent and therefore $S$ has to verify at least $n / 4 c^2$ records.
\end{proof}

\section{Proof of Theorem \ref{thm:wcont}} \label{sec:app:proofOfLipschitz}

 \begin{proof}
      We set $p_i = \frac{2w_i}{3 f(\xw) \eps}$ and we show that 
    those values satisfy the LP \eqref{eq:neccLP1}. Thus, if we choose to verify record $i$ with probability $\min\{2p_i, 1\}$, we 
    get a valid $(\varepsilon,\delta)$-certification scheme.

    For any subset $S \subseteq \Workers$ by the  $\vec{w}$-Lipschitz property we get that
    \[ |f(\xw) - f(\vec{x}_{\Workers\setminus S})| \le \sum_{i \in S} w_i \Leftrightarrow \left| \frac{2}{3 \eps} - \frac{2f(\vec{x}_{\Workers\setminus S})}{3 \eps f(\xw)} \right| \le \sum_{i \in S} \frac{2w_i}{3 f(\xw) \eps}. \]

    \noindent Now if $\frac{f(\xw)}{f(\vec{x}_{\Workers\setminus S})} > \frac{1}{1 - \eps}$, we have
    $ \left| \frac{2}{3 \eps} - \frac{2f(\vec{x}_{\Workers\setminus S})}{3 \eps f(\xw)} \right| > \frac 2 3. $

    \noindent Also if $\frac{f(\xw)}{f(\vec{x}_{\Workers\setminus S})} < 1 - \eps$, we have
    $ \left| \frac{2}{3 \eps} - \frac{2f(\vec{x}_{\Workers\setminus S})}{3 \eps f(\xw)} \right| >
    \frac{2}{3 \eps (1-\eps)} - \frac{2}{3 \eps}  \ge \frac 2 3. $
    
    \noindent Therefore when $\frac{f(\xw)}{f(\vec{x}_{\Workers\setminus S})} \notin \left[1 - \eps, \frac{1}{1 - \eps}\right]$, we have
    $\sum_{i \in S} p_i = \sum_{i \in S} \frac{2w_i}{3 f(\xw) \eps} \ge \left| \frac{2}{3 \eps} - \frac{2f(\vec{x}_{\Workers\setminus S})}{3 \eps f(\xw)} \right| \ge \frac{2}{3}$.

    \noindent This means that LP \eqref{eq:neccLP1} is satisfied. 
    Now we can apply Theorem \ref{thm:optInstV} and we conclude that the certification scheme that verifies each record independently 
    with probability $\min\{2p_i, 1\}$, where $2p_i = \frac{4w_i}{3 f(\xw) \eps}$, verifies at most 
    $\frac{4\sum_{i \in \Workers} w_i}{3 f(\xw) \eps}$ records and has probability of success at least $2/3$. In order to get 
    probability of success $\delta$ we instead verify each record $i$ with probability $2 p_i \log(1/\delta)$ and the theorem follows.
  \end{proof}

\section{Complete Statement and applications of Theorem \ref{thm:sCorrection}} \label{sec:app:applicationsStrongCorrection}

  More precisely we are given an input $\xw = (x_1, x_2 , \cdots , x_n)$ of length $n$, where every $x_i$ belongs in some set 
$\Domain$. In this section, we will fix $\Domain = [\mathcal{D}]^d$ for some $\mathcal{D} = n^{O(1)}$ to be the discretized 
$d$-dimensional Euclidean space. Our goal is to compute the value of a symmetric function 
$f : \Domain^n \rightarrow \mathbb{R}_+$ with input $\vec x \in \Domain^n$. We assume that all $x_i$ are distinct and define 
$\Support \subseteq \Domain$ as the set $\Support = \{x_i : i \in \Workers\}$. Since we consider symmetric functions $f$, it 
is convenient to extend the definition of $f$ to sets $f(\Support) = f(x)$.

  The \emph{conditional sampling model} allows such queries of small description complexity to be performed. In particular, 
the algorithm is given access to an oracle $\Cond(C)$ that takes as input a function $C: \Domain \rightarrow \{0,1\}$ and 
returns a tuple $(i, x_i)$ with $C(x_i) = 1$ with $i$ chosen uniformly at random from the subset 
$\{ j \in [n] \mid C(x_j) = 1 \}$. If no such tuple exists the oracle returns $\bot$.

  The main result of this section is a reduction from any algorithm that uses conditional sampling to a strong correction
scheme.

\begin{theorem}
    An algorithm that uses $k$ conditional samples to compute a function $f$ can produce a strong correction scheme with 
  verification cost $k$.
\end{theorem}

\begin{proof}
    We will show how we can implement one conditional sample using only one verification. We take all the values of the
  records $x_1, \dots, x_n$ and we randomly shuffle them to get $x_{\pi_1}, \dots, x_{\pi_n}$. Then we take one by one 
  the records $x_{\pi_i}$ with this new order and we check if $C(x_{\pi_i}) = 1$. If yes then we verify $x_{\pi_i}$ and
  if it is valid we return it as the result of the conditional sampling oracle. If $\pi_i$ is invalid then we just ignore
  this records without any cost and we proceed with the next record. If we finish the records and we found no valid record 
  $x_{\pi_j}$ such that $C(x_{\pi_j}) = 1$, then we return $\bot$. It is easy to see that this procedure produces at every 
  step a verified conditional sample. Since the conditional sampling algorithm has only this access to the data we get that
  any guarantees of the conditional sampling immediately transfer to this corresponding strong correction scheme.
\end{proof}

  The above result gives a general framework for designing strong correction schemes for several computational and learning 
problems. We give some of these examples below that are based on the work of \cite{GouleakisTZ2017}. For other distributional 
learning tasks, one can use the conditional sampling algorithms of \cite{CanonneRS14} to get efficient strong correction 
schemes. Some applications of Theorem \ref{thm:sCorrection} can be found in Appendix \ref{sec:app:applicationsStrongCorrection}.
\paragraph{$k$-means Clustering}
  Let $\Domain$ be a metric space with distance metric $d : \Domain \times \Domain \rightarrow \reals$, i.e. $d(x, y)$ 
represents the distance between $x$ and $y$. Given a set of \textit{centers} $Ct$ we define the distance of a point $x$ from 
$Ct$ to be $d(x, Ct) = \min_{c \in Ct} d(x, c)$. Now given a set of $n$ input points $\Support \subseteq \Domain$ and a set of 
centers $Ct \subseteq \Omega$ we define the cost of $Ct$ for $\Support$ to be 
$d(\Support, Ct) = \sum_{x \in \Support} d(x, Ct)$. The $k$-means problem is the problem of minimizing the 
\textit{squared cost} $d^2(\Support, Ct) = \sum_{x \in \Support} d^2(x, Ct)$ over the choice of centers $Ct$ subject to the 
constraint $|Ct| = k$. We assume that the diameter of the metric space is $\Delta = \max_{x, y \in \Support} d(x, y)$. In this 
setting we assume that the records contain the points in the $d$-dimensional metric space.

\begin{corollary} \label{cor:kmeans}
    Let $x_1, x_2, \dots, x_n$ be the points in the $d$-dimensional metric space $\Domain$ stored in the records $\Workers$
  and $f(\xw)$ be the optimal $k$-means clustering of the points $\xw$. There exists a strong correction scheme with 
  $\tilde{O}(k^2 \log n \log ( k / \delta ))$ verifications that guarantees a constant approximation of the value optimal 
  clustering, with probability of failure at most $\delta$.
\end{corollary}

  The proof of this corollary is based on the Theorem \ref{thm:sCorrection} and the Theorem 2 from \cite{GouleakisTZ2017}.

\paragraph{Euclidean Minimum Spanning Tree}
  Given a set of points $\xw$ in $\reals^d$, the minimum spanning tree problem in Euclidean space ask to compute the a spanning
tree $T$ on the points minimizing the sum of weights of the edges. The weight of an edge between two points is equal to their 
Euclidean distance. We will focus on a simpler variant of the problem which is to compute just the weight of the best possible 
spanning tree, i.e. estimate the quantity
$\min_{\text{tree } T} \sum_{(x,x') \in T} \|x - x'\|_2$.

\begin{corollary} \label{cor:mst}
    Let $x_1, x_2, \dots, x_n$ be the points in $\reals^d$ stored in the records $\Workers$ and 
  $f(\xw) = \min_{\text{tree } T} \sum_{(x,x') \in T} \|x - x'\|_2$. There exists a strong correction scheme with 
  $\tilde O(d^3 \log^4 n / \eps^7 )\cdot \log(1/\delta)$ verifications that guarantees an $(1 + \eps)$-approximation of
  the weight of the minimum spanning tree, with probability of failure at most $\delta$.
\end{corollary}

  The proof of this corollary is based on the Theorem \ref{thm:sCorrection} and the Theorem 3 from \cite{GouleakisTZ2017}.

\textbf{Remark.} Observe that the value of the MST gives a 2-approximation of the metric TSP and the metric Steiner Tree 
problems. Hence Corollary \ref{cor:mst} implies efficient strong correction schemes that achieve constant approximation for 
those problems as well.
 
\end{document}